\documentclass{article}
\usepackage[utf8]{inputenc}
\usepackage[margin=1in]{geometry}
\usepackage{lmodern}
\usepackage{mathpazo}

\usepackage{graphicx}
\usepackage{xcolor}
\usepackage{amsmath,amssymb,amsfonts,amsthm,url,hyperref}
\usepackage{cleveref}

\newtheorem{theorem}{Theorem}

\newtheorem{lemma}[theorem]{Lemma}

\newcommand{\N}{\mathbb{N}}
\newcommand{\Z}{\mathbb{Z}}
\newcommand{\bA}{\mathbf{A}}
\newcommand{\bG}{\mathbf{G}}
\newcommand{\bB}{\mathbf{B}}

\newcommand{\bI}{\mathbf{I}}
\newcommand{\bzero}{\mathbf{0}}
\newcommand{\bone}{\mathbf{1}}
\newcommand{\R}{\mathbb{R}}
\newcommand{\bw}{\mathbf{w}}
\newcommand{\be}{\mathbf{e}}
\newcommand{\bu}{\mathbf{u}}

\newcommand{\bx}{\mathbf{x}}

\newcommand{\bz}{\mathbf{z}}
\newcommand{\bxs}{\mathbf{x}^{\text{sum}}}
\newcommand{\bws}{\mathbf{w}^{\text{sum}}}
\newcommand{\qnaesat}{NAE$\forall\exists$3SAT}
\newcommand{\naesat}{NAE3SAT}
\newcommand{\eps}{\epsilon}
\DeclareMathOperator{\lindisc}{lindisc}
\DeclareMathOperator{\poly}{poly}

\title{Linear Discrepancy is $\Pi_2$-Hard to Approximate}
\author{Pasin Manurangsi\thanks{Google Research. Email: pasin@google.com.}}
\date{\today}

\begin{document}

\maketitle

\begin{abstract}
In this note, we prove that the problem of computing the linear discrepancy of a given matrix is $\Pi_2$-hard, even to approximate within $9/8 - \eps$ factor for any $\eps > 0$. This strengthens the NP-hardness result of Li and Nikolov~\cite{LiN20} for the exact version of the problem, and answers a question posed by them. Furthermore, since Li and Nikolov showed that the problem is contained in $\Pi_2$, our result makes linear discrepancy another natural problem that is $\Pi_2$-complete (to approximate).
\end{abstract}

\section{Introduction}

The \emph{linear discrepancy}~\cite{LovaszSV86} of a matrix $\bA \in \R^{m \times n}$ is defined as 
\begin{align*}
\lindisc(\bA) := \max_{\bw \in [0, 1]^n} \min_{\bx \in \{0, 1\}^n} \|\bA(\bw - \bx)\|_{\infty}. 
\end{align*}

Besides its connection to combinatorial discrepancy and its variants~\cite{LovaszSV86}, the linear discrepancy also has applications in other areas, such as approximation algorithms (e.g.~\cite{EisenbrandPR13,Rothvoss16,HobergR17}); for example, the best known approximation algorithm for the bin packing problem~\cite{HobergR17} uses an algorithm for linear discrepancy of~\cite{LovettM15} as a subroutine. Although such an algorithm can (given $\bA$ and $\bw$) find $\bx$ that has ``small'' $\|\bA(\bx - \bw)\|_{\infty}$, the guaranteed bound does not directly involve $\lindisc(\bA)$. Thus, it is not an (approximation) algorithm for computing $\lindisc(\bA)$ given $\bA$. Li and Nikolov~\cite{LiN20} investigated the complexity of this problem; they showed that the problem belongs to $\Pi_2$ and that it is NP-hard. Due to this, they asked whether the problem is $\Pi_2$-hard (which would mean that it is complete for $\Pi_2$). We resolve this problem by showing that approximating linear discrepancy is $\Pi_2$-hard:

\begin{theorem} \label{thm:main}
It is $\Pi_2$-hard to approximate the linear discrepancy of a given matrix within $\frac{9}{8} - \eps$ factor for any $\eps > 0$.
\end{theorem}

\paragraph{Other Related Works.}
The computational complexity of computing or approximating several notions of discrepancy has been investigated in recent years. The aforementioned work of Li and Nikolov~\cite{LiN20} gave several algorithms for $\lindisc$ when $m$ is small, and a polynomial time $O(2^n)$-approximation algorithm in the general case. For the classic notion of (combinatorial) discrepancy, Charikar et al.~\cite{CharikarNN11} showed that it is NP-hard to distinguish even the case that the discrepancy is zero and the case where the discrepancy is $\Omega(\sqrt{m})$ when $m = O(n)$; this is essentially the strongest possible hardness for the problem since the algorithms of \cite{Bansal10,LovettM15} ensures $O(\sqrt{m})$ bound for any $\bA$ with $m = O(n)$. For the hereditary discrepancy, Matoušek et al.~\cite{MVT20} gave a polylogarithmic approximation algorithm for the problem while Austrin et al.~\cite{AustrinGH17} showed that it is NP-hard to approximate beyond a factor of 2.

As pointed out in~\cite{LiN20}, linear discrepancy is also related to the covering radius of a lattice, which can be defined in a similar manner as linear discrepancy except $\bx$ is over all $\Z^n$ (instead of $\{0, 1\}^n$). For this problem, Haviv and Regev~\cite{HavivR12} showed that it is $\Pi_2$-hard to approximate within some constant factor. Another related problem is that of computing the covering radius of a linear error correcting code; this problem is known to be $\Pi_2$-hard to approximate within some constant factor and NP-hard to approximate within $\Omega(\log \log n)$ factor~\cite{GuruswamiMR04}. Since both of these problems belong to $\Pi_2$, they are examples of problems which are $\Pi_2$-complete to approximate within some constant factor; with our main result, computing linear discrepancy now joins this class of problems.

\section{Notations}

For $k \in \N$, we use $[k]$ as a shorthand for $\{1, \dots, k\}$. We use $\bI_k$ to denote the $(k \times k)$ identity matrix. We use $\bone_k$ (resp. $\bzero_k$) to denote the $k$-dimensional all-ones (resp. all-zeros) vector. Similarly, we use $\bone_{k \times k'}$ (resp. $\bzero_{k \times k'}$) to denote the $(k \times k')$ all-ones (resp. all-zeros) matrix. When the dimensions are clear from context, we may discard the subscript and simply write $\bone$ or $\bzero$. We let $\be^{(i)}$ denote the $i$-th element of the standard basis, i.e. the vector whose $i$-th entry is one and all other entries are zero.

For $\bA \in \R^{m \times n}$ and $\bw \in [0, 1]^n$, let $\lindisc(\bA, \bw)$ denote $\min_{\bx \in \{0, 1\}^n} \|\bA(\bw - \bx)\|_{\infty}$. Note that this means that $\lindisc(\bA) = \max_{\bw \in [0, 1]^n} \lindisc(\bA, \bw)$; each maximizer $\bw$ is said to be a \emph{deep hole} of $\bA$.

\section{Warm-up: NP-hardness of Approximating $\lindisc$}

In this section, we will prove a weaker version of \Cref{thm:main} in which $\Pi_2$-hardness is relaxed to only NP-hardness, as stated below. While this result is of course subsumed by \Cref{thm:main}, it demonstrates the main new gadget required in our work, which will also be used for the $\Pi_2$-hardness result.

\begin{theorem} \label{thm:np-hard}
It is NP-hard to approximate the linear discrepancy of a given matrix within $\frac{9}{8} - \eps$ factor for any $\eps > 0$.
\end{theorem}

Following~\cite{LiN20}, we reduce from the \naesat\ problem where we are given a set $V$ of variables and a 3CNF formula $\phi$ over $V$. The goal is to determine whether there exists an assignment to $V$ which makes every clause of $\phi$ has at least one literal evaluated to true and at least one literal evaluated to false.

We will now give an informal intuition for our proof. We will sometimes be vague; everything will be formalized below. Our reduction builds on the reduction of~\cite{LiN20}, which works by viewing each clause as a row vector of the matrix $\bA$ in a natural manner, i.e., the $i$-th entry is 1 if the literal $x_i$ is present, -1 if the literal $\neg x_i$ is present and 0 otherwise. It is not hard to see that, when the starting instance is a NO instance of the \naesat problem, then $\lindisc(\bA, 0.5 \cdot \bone) \geq 3/2$. Similarly, in the YES case, $\lindisc(\bA, 0.5 \cdot \bone) \leq 1/2$. The latter unfortunately is insufficient to conclude that $\lindisc(\bA)$ is small, as $0.5 \cdot \bone$ may not be a deep hole of $\bA$. Nonetheless, Li and Nikolov managed to use the fact that a deep hole has a polynomial bit complexity to prove that $\lindisc(\bA, 0.5 \cdot \bone) \leq 3/2 - \exp(-\poly(nm))$ in the YES case. Thus, their result only implies hardness of approximation with factor only $1 + \exp(-\poly(nm))$.

As one can see from the above outline, the most important challenge in the above reduction is in bounding $\lindisc(\bA, \bw)$ for $\bw \ne 0.5 \cdot \bone$ in the YES case. Our main idea is to make multiple (specifically 3) copies of each column. We then add gadgets on them in such a way that, when $\bw = 0.5 \cdot \bone$, the three columns are forced to have the same value (and hence the NO case remains similar to before). In the YES case, we can show that, while our gadget is restrictive when $\bw = 0.5 \cdot \bone$, it is ``not as restrictive'' for $\bw$ far from $0.5 \cdot \bone$, which eventually allows us to overcome the previously challenging scenario.

\subsection{Our Gadget}
Our gadget is a simple $(3 \times 3)$ matrix $\bG$ that can be used to enforce the three columns to be the same when the target vector $\bu$ is $0.5 \cdot \bone$, which will be used in the NO (i.e. soundness) case. On the other hand, for any target vector $\bu$ (not necessarily equal to $0.5 \cdot \bone$) and any sign $b \in \{-1, +1\}$, we can find a 0-1 vector $\bz$ with low discrepancy with respect to $G$ while also maintaining that the sign of the sum of entries of $\bu - \bz$ agrees with $b$. And that such a sum does not have too large absolute value. We note that, in the reduction, the sign $b$ will be selected according to the assignment of the \naesat\ instance, and that the sign agreement together with the absolute value bound help ensure that the linear discrepancy is small in the YES case.

Our gadget's properties are formalized below.


\begin{lemma} \label{lem:gadget-half}
Let $\bG = \begin{bmatrix} 1 & 1 & -1 \\ 1 & -1 & 1 \\ -1 & 1 & 1 \end{bmatrix}$. Then, the following holds:
\begin{itemize}
\item (Completeness) For any $\bu \in [0, 1]^3$ and $b \in \{-1, 1\}$, there exists $\bz \in \{0, 1\}^3$ such that 
\begin{itemize}
\item (Low Discrepancy w.r.t. $\bG$) $\|\bG(\bu - \bz)\|_{\infty} \leq 4/3$.
\item (Sign Agreement) $b \cdot (\bone^T (\bu - \bz)) \geq 0$
\item (Low Discrepancy w.r.t. $\bone^T$) $|\bone^T (\bu - \bz)| \leq 2.$
\end{itemize}
\item (Soundness) If $\bz \in \{0, 1\}^3$ such that $\|\bG(0.5 \cdot \bone - \bz)\|_\infty < 3/2$, then $\bz$ is either $\bone$ or $\bzero$.
\end{itemize}
\end{lemma}

\begin{proof}

\textbf{(Completeness)} Due to symmetry, we may assume w.l.o.g. that $b = 1$ and that $\bu_1 \leq \bu_2 \leq \bu_3$. We then consider three cases as follows:

\begin{itemize}
\item Case I: $\bu_1 + \bu_2 + \bu_3 \geq 2$. In this case, pick $\bz_1 = 0, \bz_2 = \bz_3 = 1$. Let us now verify the three desired properties below.
\begin{itemize}
\item (Low Discrepancy w.r.t. $\bG$) We have
\begin{align*}
\bG(\bu - \bz) =
\begin{bmatrix}
\bu_1 + \bu_2 - \bu_3 \\
\bu_1 - \bu_2 + \bu_3 \\
-\bu_1 + \bu_2 + \bu_3 - 2
\end{bmatrix}
\end{align*}
Let us now bound each entry of the above vector. For the first entry, we have
\begin{align*}
-1 \leq -\bu_3 \leq \bu_1 + \bu_2 - \bu_3 \leq \bu_1 \leq 1.
\end{align*}
For the second entry, we have
\begin{align*}
-1 \leq -\bu_2 \leq \bu_1 - \bu_2 + \bu_3 \leq \bu_2 \leq 1.
\end{align*}
For the third entry, we can upper bound it by
\begin{align*}
-\bu_1 + \bu_2 + \bu_3 - 2 \leq 0 + 1 + 1 - 2 = 0.
\end{align*}
For the lower bound of the third entry, we have
\begin{align*}
-\bu_1 + \bu_2 + \bu_3 - 2 \geq \bu_3 - 2 \geq 2/3 - 2 = -4/3.
\end{align*}
where the second inequality follows from $\bu_3 \geq \bu_2 \geq \bu_1$ and $\bu_1 + \bu_2 + \bu_3 \geq 2$.

As a result, we have $\|\bG(\bu - \bz)\|_\infty \leq 4/3$.

\item (Sign Agreement) From our assumption $\bu_1 + \bu_2 + \bu_3 \geq 2$, we have
\begin{align*}
\bone^T (\bu - \bz) = \bu_1 + \bu_2 + \bu_3 - 2 \geq 0.
\end{align*}
\item (Low Discrepancy w.r.t. $\bone^T$)
Since $\bu_1, \bu_2, \bu_3 \leq 1$, we have
\begin{align*}
\bone^T (\bu - \bz) = \bu_1 + \bu_2 + \bu_3 - 2 \leq 3 - 2 = 1.
\end{align*}
From this and the sign agreement shown above, we have $|\bone^T (\bu - \bz)| \leq 1$.
\end{itemize}

\item Case II: $\bu_1 + \bu_2 + \bu_3 < 2$ and $-\bu_1 + \bu_2 + \bu_3 \leq 4/3$. In this case, pick $\bz_1 = \bz_2 = \bz_3 = 0$. Let us now verify the three desired properties below.
\begin{itemize}
\item (Low Discrepancy w.r.t. $\bG$) We have
\begin{align*}
\bG(\bu - \bz) =
\begin{bmatrix}
\bu_1 + \bu_2 - \bu_3 \\
\bu_1 - \bu_2 + \bu_3 \\
-\bu_1 + \bu_2 + \bu_3
\end{bmatrix}
\end{align*}
Let us now bound each entry of the above vector. The first two entries can be bounded in the same matter as in Case I. For the third entry, we have
\begin{align*}
0 \leq \bu_3 \leq -\bu_1 + \bu_2 + \bu_3 \leq 4/3,
\end{align*}
where the right most inequality follows from the the second assumption in this case.

As a result, we have $\|\bG(\bu - \bz)\|_\infty < 4/3$.
\item (Sign Agreement) We have
\begin{align*}
\bone^T (\bu - \bz) = \bu_1 + \bu_2 + \bu_3 \geq 0.
\end{align*}
\item (Low Discrepancy w.r.t. $\bone^T$)
From our assumption $\bu_1 + \bu_2 + \bu_3 < 2$, we have
\begin{align*}
\bone^T (\bu - \bz) = \bu_1 + \bu_2 + \bu_3 < 2.
\end{align*}
\end{itemize}

\item Case III: $\bu_1 + \bu_2 + \bu_3 < 2$ and $-\bu_1 + \bu_2 + \bu_3 > 4/3$. In this case, pick $\bz_1 = \bz_2 = 0, \bz_3 = 1$. Let us now verify the three desired properties below.
\begin{itemize}
\item (Low Discrepancy w.r.t. $\bG$) We have
\begin{align*}
\bG(\bu - \bz) =
\begin{bmatrix}
\bu_1 + \bu_2 - \bu_3 + 1 \\
\bu_1 - \bu_2 + \bu_3 - 1 \\
-\bu_1 + \bu_2 + \bu_3 - 1
\end{bmatrix}
\end{align*}

Let us now bound each entry of the above vector. 
For the first entry, we have
\begin{align*}
0 \leq -\bu_3 + 1 \leq \bu_1 + \bu_2 - \bu_3 + 1  \leq \bu_1 + 1 = \frac{1}{2}\left((\bu_1 + \bu_2 + \bu_3) - (-\bu_1 + \bu_2 + \bu_3)\right) + 1 < 4/3,
\end{align*}
where the right most inequality follows from the two assumptions of this case.
For the second entry, we have
\begin{align*}
-1 \leq \bu_1 - 1 \leq \bu_1 - \bu_2 + \bu_3 - 1 \leq \bu_3 - 1 \leq 0.
\end{align*}
For the third entry, we similarly have
\begin{align*}
-1 \leq \bu_3 - 1 \leq -\bu_1 + \bu_2 + \bu_3 - 1 \leq \bu_1 + \bu_2 + \bu_3 - 1 < 1,
\end{align*}
where the last inequality follows from the assumption $\bu_1 + \bu_2 + \bu_3 < 2$.
As a result, we have $\|\bG(\bu - \bz)\|_\infty < 4/3$.

\item (Sign Agreement) From $-\bu_1 + \bu_2 + \bu_3 > 4/3$, we have
\begin{align*}
\bone^T (\bu - \bz) = \bu_1 + \bu_2 + \bu_3 - 1 \geq -\bu_1 + \bu_2 + \bu_3 - 1 > 0.
\end{align*}
\item (Low Discrepancy w.r.t. $\bone^T$)
From our assumption $\bu_1 + \bu_2 + \bu_3 < 2$, we have
\begin{align*}
\bone^T (\bu - \bz) = \bu_1 + \bu_2 + \bu_3 - 1 < 1.
\end{align*}

\end{itemize}
\end{itemize}
In all cases, we have found a desired $\bz$. This concludes the proof of completeness.

\paragraph{(Soundness)} Consider any $\bz \in \{0, 1\}^3$ such that $\bz \ne \bone, \bzero$. Due to symmetry, we may w.l.o.g. assume $\bz = \begin{bmatrix} 1 \\ 1 \\ 0 \end{bmatrix}$. In this case, we have
\begin{align*}
\|\bG(0.5 \cdot \bone - \bz)\|_{\infty} = \left\|\begin{bmatrix} -1.5 \\ 0.5 \\ 0.5 \end{bmatrix}\right\|_{\infty} \geq 1.5,
\end{align*} 
which concludes our proof.
\end{proof}

\subsection{The Reduction}

Having described our gadget, we will now proceed to the reduction, whose properties are summarized below in \Cref{lem:np-reduction}. Note that this, together with NP-hardness of \naesat\ (e.g.~\cite{Schaefer78}), implies \Cref{thm:np-hard}.

\begin{lemma} \label{lem:np-reduction}
There exists a polynomial-time reduction that takes in an instance $(V, \phi)$ of \naesat\ and produces a matrix $\bA$ such that the following holds.
\begin{itemize}
\item (Completeness) If $(V, \phi)$ is a YES instance of \naesat, then $\lindisc(\bA) \leq 4/3$,
\item (Soundness) If $(V, \phi)$ is a NO instance of \naesat, then $\lindisc(\bA) \geq 3/2$.
\end{itemize}
\end{lemma}

\begin{proof}
Let the variables in $V$ be $v_1, \dots, v_n$, and let the clauses in $\phi$ be $C_1, \dots, C_m$. We first create a matrix $\bB \in \R^{m \times n}$ as follows. For every $j \in [m]$, suppose that $C_j$ contains the literals\footnote{We assume w.l.o.g. that it contains exactly three literals; otherwise, we can just replicate one of the literals.} $b_1 v_{i_1}$, $b_2 v_{i_2}$ and $b_3 v_{i_3}$, where $b_1, b_2, b_3 \in \{-1, 1\}$ indicate whether the literals are negated. We let the $j$-th row of $\bB$ be $b_1 \be^{(i_1)} + b_2 \be^{(i_2)} + b_3 \be^{(i_3)}$. Then, we let $\bA \in \R^{(m + 3n) \times 3n}$ be defined as
\begin{align*}
\bA
&=
\begin{bmatrix}
\frac{1}{3}\bB & \frac{1}{3}\bB & \frac{1}{3}\bB \\
\bI & \bI & -\bI \\
\bI & -\bI & \bI \\
-\bI & \bI & \bI
\end{bmatrix}.
\end{align*}

It is obvious that the reduction runs in polynomial time. We will now prove the completeness and soundness of the reduction. To prove these, for every $\bx \in \{0, 1\}^{3n}$, we view it as a concatenation of $\bx^1, \bx^2, \bx^3$ each of length $n$. Similarly, we view each $\bw \in [0, 1]^{3n}$ as a concatenation of $\bw^1, \bw^2, \bw^3$ each of length $n$. Furthermore, let $\bxs = \bx^1 + \bx^2 + \bx^3$ and $\bws = \bw^1 + \bw^2 + \bw^3$.

The following identity will be useful in the subsequent steps of the proof:
\begin{align} \label{eq:inf-decompose-np}
\|\bA(\bw - \bx)\|_{\infty} = \max\left\{\frac{1}{3} \|\bB(\bws - \bxs)\|_{\infty}, \max_{i \in [n]}\left\{\left\|\bG\left(\begin{bmatrix} \bw^1_i \\ \bw^2_i \\ \bw^3_i \end{bmatrix} - \begin{bmatrix} \bx^1_i \\ \bx^2_i \\ \bx^3_i \end{bmatrix}\right)\right\|_{\infty}\right\}\right\}
\end{align}
where $\bG$ is the matrix in Lemma~\ref{lem:gadget-half}

\paragraph{(Completeness)} Suppose that $(V, \phi)$ is a YES instance. Consider any $\bw \in [0, 1]^n$; we will show that $\lindisc(\bA, \bw) \leq 4/3$. Since $(V, \phi)$ is a YES instance, there exists an assignment $\psi: V \to \{0, 1\}$ that assigns at least one literal to false and one literal to true in each clause. We construct our $\bx$ by letting $\begin{bmatrix} \bx^1_i \\ \bx^2_i \\ \bx^3_i \end{bmatrix}$ be the vector $\bz$ from the completeness of Lemma~\ref{lem:gadget-half} with $\bu = \begin{bmatrix} \bw^1_i \\ \bw^2_i \\ \bw^3_i \end{bmatrix}$ and the sign $b = 1 - 2\psi(v_i)$. From our choice of $\bx$ and the first property of the completeness of Lemma~\ref{lem:gadget-half}, we have
\begin{align*}
\left\|\bG\left(\begin{bmatrix} \bw^1_i \\ \bw^2_i \\ \bw^3_i \end{bmatrix} - \begin{bmatrix} \bx^1_i \\ \bx^2_i \\ \bx^3_i \end{bmatrix}\right)\right\|_{\infty} \leq 4/3.
\end{align*}
Hence, by~\eqref{eq:inf-decompose-np}, we are left to only show that $\|\bB(\bws - \bxs)\|_{\infty} \leq 4$. To do this, observe that the second property of Lemma~\ref{lem:gadget-half} can be written as 
\begin{align} \label{eq:sgn-agreement-assignment}
(1 - 2\psi(v_i)) \cdot (\bws_i - \bxs_i) \geq 0
\end{align} and the third property can be written as 
\begin{align} \label{eq:abs-bound-assignment}
|\bws_i - \bxs_i| \leq 2.
\end{align}
Consider the dot product of $j$-th row of $\bB$ and $(\bws - \bxs)$. It results in only three non-zero terms, and the choice of $\psi$ together with that of~\eqref{eq:sgn-agreement-assignment} ensures that at most two of these terms are positive and at most two of them are negative. Furthermore,~\eqref{eq:abs-bound-assignment} ensures that the absolute value of each term is at most 2. As a result, their sum has absolute value at most $4$. In other words, we have $\|\bB(\bws - \bxs)\|_{\infty} \leq 4$ as desired.

\paragraph{(Soundness)} Suppose contrapositively that $\lindisc(\bA) < 3/2$; we will show that $(V, \phi)$ is a YES instance of \naesat. Since $\lindisc(\bA) < 3/2$, there must exists $\bx \in \{0, 1\}^{3n}$ such that $\|\bA(0.5 \cdot \bone - \bx)\|_{\infty} < 3/2$. From this and~\eqref{eq:inf-decompose-np}, we have 
\begin{align*}
3/2 > \left\|\bG\left(0.5 \cdot \bone - \begin{bmatrix} \bx^1_i \\ \bx^2_i \\ \bx^3_i \end{bmatrix}\right)\right\|_{\infty}
\end{align*}
for all $i \in [n]$. Applying Lemma~\ref{lem:gadget-half}, we can conclude that $\bx^1 = \bx^2 = \bx^3$. This means that $\bxs = 3\bx^1$. Hence, from~\eqref{eq:inf-decompose-np}, we have
\begin{align*}
3/2 > \|\bB(0.5 \cdot \bone - \bx^1)\|_{\infty}.
\end{align*}
Let $\psi: V \to \{0, 1\}$ denote the assignment where we set $\psi(v_i) = \bx^1_i$. The above inequality implies that $\psi$ assigns at least one literal to false and one literal to true in each clause; otherwise, the corresponding row in $\bB$ when multiplied with $0.5 \cdot \bone - \bx^1$ will result in -3/2 or 3/2. Thus, $(V, \phi)$ is a YES instance as desired.
\end{proof}

\section{Proof of the Main Result: $\Pi_2$-hardness of Approximating $\lindisc$}

We will now prove our main result of the paper: $\Pi_2$-hardness of approximating $\lindisc$. To do this, we reduce from the $\Pi_2$-complete variant of \naesat, called \qnaesat. In the \qnaesat\ problem, we are given two sets $V_A, V_E$ of variables and a 3CNF formula $\phi$ over $V_A \cup V_E$. The goal is to determine whether, for every assignment to $V_A$, there exists an assignment to $V_E$ which makes every clause of 3CNF has at least one literal evaluated to true and at least one literal evaluated to false. \qnaesat\ is known to be $\Pi_2$-complete~\cite{eiter1995note}. The properties of our reduction are summarized below in \Cref{lem:pi2-reduction}, which immediately implies our main theorem (\Cref{thm:main}).

\begin{lemma} \label{lem:pi2-reduction}
There exists a polynomial-time reduction that takes in an instance $(V_A, V_E, \phi)$ of \qnaesat\ and produces a matrix $\bA'$ such that the following holds.
\begin{itemize}
\item (Completeness) If $(V_A, V_E, \phi)$ is a YES instance of \qnaesat, then $\lindisc(\bA') \leq 4/3$,
\item (Soundness) If $(V_A, V_E, \phi)$ is a NO instance of \qnaesat, then $\lindisc(\bA') \geq 3/2$.
\end{itemize}
\end{lemma}

Before we prove the lemma, let us outline the main ideas. Our reduction is in fact a minor modification of the NP-hardness reduction from the previous section. In particular, we construct the same matrix as before, and then we add $|V_A|$ additional columns (and also a certain number of appropriately constructed rows). These columns allow us to enforce the $\forall$ quantifier: selecting the entries of $\bw$ slightly above (resp. slightly below) 1/2 will force the chosen $\bx$ to be 0 (resp. 1) for those variables. 

\begin{proof}[Proof of \Cref{lem:pi2-reduction}]
Let the variables in $V_A$ be $v_1, \dots, v_{n'}$ and those in $V_E$ be $v_{n + 1}, \dots, v_n$. Furthermore, let $V = V_A \cup V_E$. We first create the matrix $\bA \in \R^{(m + 3n) \times 3n}$ as in the proof of Lemma~\ref{lem:np-reduction}. Our final matrix $\bA'$ has $n'$ additional columns and $2n'$ additional rows defined as follows:
\begin{align*}
\bA' =
& \left[\begin{array}{c|c|c|c|c|c|c}
      \multicolumn{6}{c|}{\bA} & \bzero_{n' \times n'}  \\
      \hline
      \frac{2}{3} \bI_{n'} & \bzero_{n' \times (n - n')} & \frac{2}{3} \bI_{n'} & \bzero_{n' \times (n - n')} & \frac{2}{3} \bI_{n'} & \bzero_{n' \times (n - n')} & -2\bI_{n'} \\
      \hline
      \multicolumn{6}{c|}{\bzero_{n' \times 3n}} & \frac{8}{3}\bI_{n'}  \\
    \end{array}\right]
\end{align*}

For every $\bw \in [0, 1]^{3n + n'}$, we view it as a concatenation of $\bw^1, \bw^2, \bw^3, \bw^*$ where the first three vectors have dimensions $n$ and the last vector has dimension $n'$. Similarly, we view each $\bx \in \{0, 1\}^{3n + n'}$ as a concatenation of $\bx^1, \bx^2, \bx^3, \bx^*$ where the first three vectors have dimensions $n$ and the last vector has dimension $n'$. Similar to the proof of Lemma~\ref{lem:np-reduction}, let $\bxs = \bx^1 + \bx^2 + \bx^3$ and $\bws = \bw^1 + \bw^2 + \bw^3$.

Under these notations, we may write $\|\bA'(\bw - \bx)\|_{\infty}$ as
\begin{align}
&\|\bA'(\bw - \bx)\|_{\infty} \nonumber \\
&= \max\left\{\left\|\bA\left(\begin{bmatrix} \bw^1 \\ \bw^2 \\ \bw^3 \end{bmatrix} - \begin{bmatrix} \bx^1 \\ \bx^2 \\ \bx^3 \end{bmatrix}\right)\right\|_{\infty}, \max_{i \in [n']} \left|\frac{2}{3}(\bws_i - \bxs_i) - 2(\bw^*_i - \bx^*_i)\right| , \max_{i \in [n']} \frac{8}{3}\left|(\bw^*_i - \bx^*_i)\right| \right\}. \label{eq:inf-decompose-pi2}
\end{align}

\paragraph{(Completeness)} Suppose that $(V_A, V_E, \phi)$ is a YES instance of \qnaesat. Consider any $\bw \in [0, 1]^n$; we will show that $\lindisc(\bA, \bw) \leq 4/3$. First, let $\psi_A: V_A \to \{0, 1\}$ be such that $\psi_A(v_i) = 0$ iff $\bw^*_i \leq 1/2$. Since $(V_A, V_E, \phi)$ is a YES instance, there exists $\psi_E: V_E \to \{0, 1\}$ such that $\psi_A$ and $\psi_E$ together satisfy all constraints\footnote{We say that all constraints are satisfied if every clause has at least one literal evaluated to true and at least one evaluated to false.}; let $\psi$ denote the concatenation of $\psi_A$ and $\psi_E$.  We construct $\bx^1, \bx^2, \bx^3$ in the same manner as in the completeness proof of~\Cref{lem:np-reduction} (with respect to $\psi$), and then let $\bx^*_i = \psi_A(v_i)$ for all $i \in [n']$. 

Via the same argument as in the completeness proof of~\Cref{lem:np-reduction}, we have $$\left\|\bA\left(\begin{bmatrix} \bw^1 \\ \bw^2 \\ \bw^3 \end{bmatrix} - \begin{bmatrix} \bx^1 \\ \bx^2 \\ \bx^3 \end{bmatrix}\right) \right\|_{\infty} \leq 4/3.$$ Hence, we are left to only show that the last two terms in~\eqref{eq:inf-decompose-pi2} are at most 4/3. For the last term, we also have
\begin{align*}
\frac{8}{3} |\bw^*_i - \bx^*_i| \leq \frac{8}{3} \cdot 1/2 = 4/3,
\end{align*}
where the inequality comes from our choice of $\bx^*_i = \psi_A(v_i)$

Consider the middle term of~\eqref{eq:inf-decompose-pi2}. Due to symmetry, we may w.l.o.g. consider only the case $\psi_A(v_i) = 0$. Recall the properties~\eqref{eq:sgn-agreement-assignment} and~\eqref{eq:abs-bound-assignment}. Using these, we have
\begin{align*}
\frac{2}{3}(\bws_i - \bxs_i) - 2(\bw^*_i - \bx^*_i) \overset{\eqref{eq:sgn-agreement-assignment}}{\geq} - 2(\bw^*_i - \bx^*_i) \geq -2 \cdot 1/2 = -1,
\end{align*}
and
\begin{align*}
\frac{2}{3}(\bws_i - \bxs_i) - 2(\bw^*_i - \bx^*_i) \leq \frac{2}{3}(\bws_i - \bxs_i) \overset{\eqref{eq:abs-bound-assignment}}{\leq} 4/3.
\end{align*}
Hence, we can conclude that $\|\bA'(\bw - \bx)\|_{\infty} \leq 4/3$ as desired.

\paragraph{(Soundness)} Suppose contrapositively that $\lindisc(\bA') < 3/2$; we will show that $(V_A, V_E, \phi)$ is a YES instance of \qnaesat. Consider any assignment $\psi_A: V_A \to \{0, 1\}$ to $V_A$. We will construct an assignment $\psi_E: V_E \to \{0, 1\}$ such that $\psi_A$ and $\psi_E$ together satisfies all constraints.

To do this, let $\bw^* \in \R^n$ be a vector such that
\begin{align*}
\bw^*_i =
\begin{cases}
1/3 & \text{ if } \psi_A(v_i) = 0, \\
2/3 & \text{ if } \psi_A(v_i) = 1.
\end{cases}
\end{align*}
Then, let $\bw = \begin{bmatrix} 0.5 \cdot \bone_{3n} \\ \bw^* \end{bmatrix}$. From $\lindisc(\bA') < 3/2$, there must exist $\bx \in \{0, 1\}^{3n + n'}$ such that $\|\bA'(\bw - \bx)\|_{\infty} < 3/2$. Simiar to before, let $\psi: V \to \{0, 1\}$ be defined as $\psi(v_i) = \bx^1_i$. Recall from the soundness proof of \Cref{lem:np-reduction} that $\psi$ satisfies all the constraints and furthermore $\bx^1_i = \bx^2_i = \bx^3_i$. We claim that $\psi$ is consistent with $\psi_A$. To see that this is the case, first observe that $\bx^*_i$ must be equal to $\psi_A(v_i)$ for all $i \in [n']$; otherwise, we have
\begin{align*}
\|\bA'(\bw - \bx)\|_{\infty} \overset{\eqref{eq:inf-decompose-pi2}}{\geq} \frac{8}{3} \left|(\bw^*_i - \bx^*_i)\right| \geq \frac{8}{3} \cdot \frac{2}{3} > 3/2.
\end{align*}
Next, using the middle term in~\eqref{eq:inf-decompose-pi2}, we have
\begin{align*}
3/2 &> \left|\frac{2}{3}(1.5 - \bxs_i) - 2(\bw^*_i - \bx^*_i)\right| \\
&= \left|\frac{2}{3}(1.5 - 3\psi(v_i)) - 2\left(\frac{1}{3} - \frac{2}{3} \psi_A(v_i)\right)\right| \\
&= \left|\frac{1}{3} + \frac{4}{3} \psi_A(v_i) - 2\psi(v_i)\right|, 
\end{align*}
which implies that $\psi_A(v_i) = \psi(v_i)$ for all $i \in [n']$.

Thus, if we let $\psi_E$ be $\psi$ restricted on $V_E$, then $\psi_A$ and $\psi_E$ together satisfies all the constraints.
\end{proof}

\section{Discussion and Open Questions}

In this note, we prove that approximating $\lindisc$ to within a factor of $9/8 - \eps$ is $\Pi_2$-hard. As stated above, the best known polynomial time algorithm only gives an approximation ratio of $O(2^n)$~\cite{LiN20}. It remains an interesting open question to close this gap. Specifically, a concrete direction is to prove NP-hardness or $\Pi_2$-hardness of approximation for all constant factors. For the latter, it should be noted that, for the related problem of computing the covering radius of a lattice or a linear code, approximating it to within a factor of 2 belongs to the class AM~\cite{GuruswamiMR04} and is thus unlikely to be $\Pi_2$-hard. However, we are not aware of any similar barrier for approximating linear discrepancy.

\bibliographystyle{alpha}
\bibliography{refs}

\end{document}